\documentclass[conference]{IEEEtran}
\usepackage{cite}
\usepackage{amsmath,amssymb,amsfonts}
\usepackage{algorithmic}
\usepackage{graphicx}
\usepackage{textcomp}
\usepackage{xcolor}
\usepackage{enumitem,hyperref}
\usepackage[normalem]{ulem}
\usepackage{amsthm}

\newcommand{\crr}[1]{\tau^{cr}_{#1}}
\newcommand{\nr}[1]{\tau^{nr}_{#1}}
\newcommand{\cn}[1]{\tau^{cn}_{#1}}
\newcommand{\nn}[1]{\tau^{nn}_{#1}}

\usepackage{tikz}
\usetikzlibrary{automata,arrows,positioning,calc}

\def\BibTeX{{\rm B\kern-.05em{\sc i\kern-.025em b}\kern-.08em
    T\kern-.1667em\lower.7ex\hbox{E}\kern-.125emX}}

\IEEEoverridecommandlockouts

\begin{document}

\title{Technical Report: Modeling Average False Positive Rates of Recycling Bloom Filters
\thanks{This material is based upon work supported by the National Science Foundation under Grant Nos. CNS-1910138, CNS-2106197, and CNS-2148275.  Any opinions, findings, and conclusions or recommendations expressed in this material are those of the author(s) and do not necessarily reflect the views of the National Science Foundation.}} 

\author{\IEEEauthorblockN{Kahlil Dozier}
\IEEEauthorblockA{\textit{Department of Computer Science} \\
\textit{Columbia University}\\
New York, US \\
kad2219@columbia.edu}
\and
\IEEEauthorblockN{Loqman Salamatian}
\IEEEauthorblockA{\textit{Department of Computer Science} \\
\textit{Columbia University}\\
New York, US\\
ls3748@columbia.edu}
\and
\IEEEauthorblockN{Dan Rubenstein}
\IEEEauthorblockA{\textit{Department of Computer Science} \\
\textit{Columbia University}\\
New York, US \\
danr@cs.columbia.edu}
}
\maketitle

\begin{abstract}
Bloom Filters are a space-efficient data structure used for the testing of membership in a set that errs only in the False Positive direction.  However, the standard analysis that measures this False Positive rate provides a form of worst case bound that is both overly conservative for the majority of network applications that utilize Bloom Filters, and reduces accuracy by not taking into account the actual state (number of bits set) of the Bloom Filter after each arrival.  In this paper, we more accurately characterize the False Positive dynamics of Bloom Filters as they are commonly used in networking applications.  In particular, network applications often utilize a Bloom Filter that ``recycles'': it repeatedly fills, and upon reaching a certain level of saturation, empties and fills again.  In this context, it makes more sense to evaluate performance using the average False Positive rate instead of the worst case bound.  We show how to efficiently compute the average False Positive rate of recycling Bloom Filter variants via renewal and Markov models.  We apply our models to both the standard Bloom Filter and a ``two-phase'' variant, verify the accuracy of our model with simulations, and find that the previous analysis' worst-case formulation leads to up to a 30\% reduction in the efficiency of Bloom Filter when applied in network applications, while two-phase overhead diminishes as the needed False Positive rate is tightened.
\end{abstract}


\section{Introduction}
\label{sec:intro}
Bloom Filters~\cite{Bloom1970space} are a time-tested, space-efficient data structure for identifying duplicate items within an input sequence, and have been applied in an extremely broad set of computing applications, including  packet processing and forwarding/routing in P2P networks \cite{5751342}, cache summarization and cache filtering in CDNs \cite{maggs2015algorithmic}, network monitoring \cite{li2016flowradar}, 
data synchronization \cite{7737013}, and even Biometric authentication \cite{martiri2017biometric}.   The Bloom Filter (BF) is attractive because it errs only in the {\em False Positive} direction (an arriving input, which we call a {\em message}, can be incorrectly identified as a repeat of a previous message) and never the {\em False Negative} direction (a repeat message is never incorrectly classified as new).

There are several variants of ``back-of-the-envelope'' analyses\footnote{Later in the paper, we look at different variants.} that give the False Positive rate $f$ of an $n+1$st message after $n$ messages have been inserted, e.g., one common variant yields $f = (1-(1-1/M)^{kn})^k$, where $M$ is the memory (number of bits in the BF) and $k$ is the number of hash functions, with each hash function mapping the message to one of the $M$ bits.  This analysis, while useful to demonstrate efficacy of the BF, does not measure False Positives in a way that best serves practical network applications for two reasons:

{\bf Observability}: First, consider the addition of an $n$th element.  If the application maintains a count, $b$, of the number of bits set, the next new arrival being classified as a repeat (a False Positive) can be directly computed as  $(b/M)^k$; the earlier analysis disregards the ability to inspect current BF state.  

{\bf Average rate}: Second, the False Positive rate of a newly arriving message grows as the BF's number of set bits increases.  Because of this, applications that merely wish to bound the {\em average} False Positive rate across all newly arriving messages significantly overestimate this average rate when only considering the message with the worst case bound. 

A class of applications that would benefit from a measure of False Positive rate that takes into account both observability and average rate are those that employ what we call a {\em Recycling Bloom Filter (RBF)}~\cite{maggs2015algorithmic,trindade2011hran,mchale2014stochastic}. 
 These applications utilize BF's across sufficiently long timescales, where new messages are continuously arriving, but are interspersed with repeats of previous messages.    Since newly arriving messages almost always set bits in the BF, over time the number of bits will become excessive, driving up the False Positive rate, such that some bits must be cleared to keep False Positive rates at a reasonable level.  While there are BF variants that do support deletion, their implementation introduces substantial overhead or requires knowledge of the specific elements destined for deletion  \cite{rothenberg2010deletable,lim2016ternary,fan2014cuckoo,quotientfilters}.  Such overhead and complexity is not suited for network applications where message arrivals have a ``fading popularity'' character, with repeat probability decreasing over time.    For applications with this type of message arrival process, one can utilize an active metric that maintains some measure of BF fullness, such as the number of messages inserted, or the number of bits set in the BF.  When this metric exceeds a threshold, the BF {\em recycles}: all bits are cleared, and the filling process begins anew (we say a new BF {\em cycle} is initiated).  This process can continue indefinitely, providing a simple and effective means for maintaining a low maximal (and average) False Positive rate of arrivals.  

In this paper, we develop recursive, quickly computable (on today's conventional hardware) analytical models that accurately capture the average False Positive rates of Recycling BFs (RBFs).  We show that using an active metric to gauge the recycle to maintain an {\em average} (as opposed to maximum) False Positive rate can offer a $30$\% or more increase in amount of messages stored.  

While the recycling process does introduce the possibility of False Negatives, the aforementioned decline in popularity of each message over time ensures that well-provisioned RBFs can make False Negative rates negligible. 
While a detailed investigation of False Negatives is beyond the scope of this paper, a commonly used solution (\cite{maggs2015algorithmic},  \cite{trindade2011hran}) to further drive down False Negative rate is to implement what we refer to as a {\em two-phase} RBF (see \ref{sec:twophase} for details) that splits the memory between alternately active and ``frozen'' BFs.  We extend our models to analyze these two-phase variants, and measure the loss in efficiency of minimizing False Positive rate as a function of memory and expected number of messages stored in the BF.

In this paper, we make the following contributions:
\begin{itemize}

\item In \S\ref{sec:background}, we provide background on Bloom Filters and their previous (worst case) analysis.

\item In \S\ref{sec:n-bounded} and \S\ref{sec:sigma-bounded}, we analytically examine the average False Positive rate of RBF variants using  via renewal and Markov Models, and extend our models to cover two-phase variants (\S\ref{sec:twophase}).

\item We verify our models using discrete-event-driven simulation in \S\ref{sec:sims}, and, using our models, contrast performance of the RBF variants as functions of memory size and desired False Positive rates.

\item The paper closes by presenting related work \S\ref{sec:related} and concluding \S\ref{sec:concl}.
\end{itemize}

Our findings show that the variant that recycles based on the number of bits set, while most complex to model to find the right parameters, is the best performing variant in that it can maintain a given False Positive rate while packing in the largest expected number of messages prior to recycling.  In contrast, the variants that recycle based on message counts to achieve an average False Positive rate are easier to model and achieve only a slightly smaller expected number of messages prior to recycling.  These average-case variants store around 30\%  more messages in expectation than the commonly used worst case variant.  We also find that implementing a two-phase variant will closely match the one-phase as the False Positive rate to be met decreases.
These results are important in that they point to users of BFs (and specifically RBFs) likely being overly conservative for their needs when using previous measures of False Positives.

\section{Background}
\label{sec:background}


\begin{table}[htbp]
\caption{Summary of Variables}
\begin{center}
\begin{tabular}{|c|c|}
\hline

\textbf{Variable} & \textbf{Description} \\
\hline
\hline
\multicolumn{2}{|c|}{Model input parameters}\\
\hline
$M$  & Size of RBF (bits) \\ 
\hline
 $k$ & Number of hash functions per message \\
\hline
$N$ & RBF recycle point (number of new message arrivals) \\ 
\hline

$\sigma$ & RBF recycle point (number of bits set) \\ 
\hline \hline
\multicolumn{2}{|c|}{Markov Model internal variables}\\
\hline
$\tau_{k}^{*}(i,j)$ & Probability next arrival takes RBF from $i$ to $j$ bits \\ & * is one of $\{cn, cr, nn, nr\}$\\ 
\hline

$\pi_{i}$ & Steady-state probability RBF contains $i$ bits set\\
\hline \hline 
\multicolumn{2}{|c|}{False Positive Variants}\\ \hline
$f_w$ & False Positive rate using worst-case  threshold (\ref{eq:wikiform})\\
\hline
$f_o$ & False Positive rate using fixed $N$ oracle lower bound (\ref{eq:oraclebound})\\
\hline
$f_a$ & f.p. rate using fixed $N$ bit-setting msgs lower bound (\ref{eq:avgN})\\
\hline
$f_{\sigma}^{1}$ & False Positive rate using $\sigma$ bound (\# bits set $\le \sigma$),\\ & one-phase RBF (\ref{longTermFalsePosRate})\\ 
\hline 
$f_{\sigma}^{2}$ & False Positive rate using $\sigma$ bound, two-phase RBF (\ref{eq:twophase})\\ 

\hline

\end{tabular}
\label{tab1}
\end{center}
\end{table}



\subsection{Bloom Filters}

A Bloom Filter~\cite{Bloom1970space} can be thought of as an array of $M$ bits indexed $0, \cdots , M-1$, all  initially set to $0$, that uses $k$ hash functions, each of which maps an arriving {\em message} to a bit in the array: the bit is then set to 1.  {\em Importantly, these $k$ hash functions' mappings are independent across messages, such that each message's assignment of bits is effectively independent from those assigned to other messages.}  This independence property is what makes the BF such a powerful abstraction, as well as amenable to mathematical analysis.

When a message arrives, the bits to which it hashes are checked before being set.  If they are all already set, the message is assumed to be a repeat of a prior message.  In the rare instance that a new message's $k$ hashes point to bits already set, it is incorrectly classified as a repeat, i.e., a {\em False Positive}.

\subsubsection{Colliding vs. Non-Colliding Hash Functions}

Conventional implementations of hash functions permit several of the $k$ hash function mappings to {\em collide} and hash to the same bit (i.e., a pseudo-random sampling with replacement).  A slightly more sophisticated hashing scheme can ensure the $k$ hashes for a message are distinct from one another (i.e., a pseudo-random sampling without replacement).  We respectively refer to these two Bloom Filter variants as {\em colliding} and {\em non-colliding}, and we analyze both types.  While we show how to incorporate these variants into our models, for values of $M$ and $k$ used in practice (typically  $M>1,000$), the colliding/non-colliding distinction has an unnoticeable impact on False Positive rate (collisions are sufficiently rare in the colliding variant).  Note that for either variant of hash functions, the independence of bits assigned to differing messages still holds.  Hence, our results are generally presented using the colliding variant.

\subsection{Traditional False Positive Analysis}

As discussed in \S\ref{sec:intro}, 
after having received $n$ unique messages, the (colliding) False Positive likelihood of the $n+1$st message is 
\vspace{-0.03in}
\begin{equation}
\label{eq:wikiform}
    f_w = [1 - (1 - \frac{1}{M})^{kn}]^{k}
\end{equation}
where, for a given $M$ and $n$,  $k \approx \frac{M}{n} \ln(2)$ will minimize the corresponding False Positive rate.
Wikipedia provides a succinct derivation of this simple formula~\cite{wikiBloom}, as well as a formula for non-colliding hash functions (in terms of Stirling numbers).  From (\ref{eq:wikiform}), the False Positive rate approaches $1$ as the number of messages $n$ increases--  an undesirable property for applications with no limit on the number of potential message insertions.  We call using such a bound in practice to decide when to recycle the BF a {\em worst case} bound (hence $f_w$) because the recycling occurs as soon as a particular message's False Positive rate exceeds $f_w$ (while the average over messages inserted will be lower).
We refer to the lifetime of a BF between recycling events as a {\em cycle} of the BF.

\subsection{Two-Phased Bloom Filter}

An RBF, upon recycling, clears all memory with respect to previous arrivals.  To mitigate the effects of this memory loss, a BF variant called a {\em Two-Phase Bloom Filter} has been used in various network applications, e.g.,~\cite{maggs2015algorithmic,trindade2011hran}.

The two-phase Bloom Filter can be thought of as {\em two} arrays with $M'$ bits each, indexed $0, ... , M'-1$, along with $k$ independent hash functions that are shared across both arrays.  Typically, $M' = \frac{M}{2}$, where $M$ is what the application's memory constraint would be for the standard (one phase) Bloom Filter.

The two-phase variant operates much like the standard recycling variant, with the addition that at any point in time, one Filter array is designated {\em active} and the other {\em frozen}.  Incoming messages are hashed and added to the active filter.  When the active filter is considered to be ``full enough'' (e.g., the number of bits set or messages stored exceeds some threshold), instead of clearing the bits of the active filter, we clear the bits of the {\em frozen} filter.  The roles of the active and frozen filter are then switched; incoming messages are added to the previously frozen (now active) filter, and the status of the previously active (now frozen) filter remains unaffected.  This process repeats indefinitely. 

In the two-phase variant, an arriving message is hashed and its bits are compared separately within both the active and frozen filters.  A match with either filter is interpreted as the message having been received previously (and having set the corresponding bits in the respective filter).  While a two-phase RBF cannot reduce its false positive rate as effectively as its one-phase counterpart using identical total memory, its corresponding False Negative rates will often be significantly lower.

\subsection{BF Thresholding Approaches}
\label{sec:thresholding}

There are two basic approaches to deciding when to recycle a BF.  One can count the number of bits $b$ that are set in the BF and recycle when $b$ exceeds a threshold $\sigma$.  We call this approach a {\em $\sigma$-bounded approach}.  Alternatively, one can bound the number of bit-setting messages, $N$, that can be admitted into the BF before recycling.\footnote{Messages that don't set bits are generally not counted because the user must assume they are repeat messages, and since repeat messages never set bits, the False Positive rate of the BF does not increase due to their repeated arrival.}  We call this approach an {\em $N$-bounded approach}.  We note that one complication with an $N$-bounded approach is that when a message that arrives without setting any bits, a user cannot tell if it is a repeat message or a False Positive.  A drawback of formulas for an $N$-bounded approach, such as (\ref{eq:wikiform}), is that they count new messages regardless of whether or not new bits are set; since in practice new messages that do not set bits are assumed to be repeats, a user will underestimate the number of new messages received when False Positives occur, such that applications of a formula such as (\ref{eq:wikiform}) to gauge when to recycle will be slightly off.

\subsection{Averaging Variants}
\label{sec:avg-variants}
A BF is only of use when the message arrival process consists both of new and repeat messages.  We wish to clarify a subtle point about what we are averaging over.  When a message first arrives, if it sets additional bits, it is identified as new, and if not, it is classified as a repeat.  Consider the $i$th new message that arrives, and let $\eta(i)$ represent the number of times the message arrives within a cycle, and let $F_i$ be an indicator that equals 1 if the $i$th message, upon first arrival, is a False Positive.  We can count these $\eta(i)$ arrivals in three different ways.
\begin{enumerate}
   
\item {\bf count-first}: We can assume that only the first arrival can count as a False Positive, making the False Positive rate for the cycle equal to $\sum_i F_i / \sum_i \eta(i)$
\item {\bf count-each:} We can assume that each arrival of a message that initially triggered a False Positive is counted as a False Positive, making the False Positive rate equal to $\sum_i F_i \eta(i) / \sum_i \eta(i)$
\item {\bf count-instance}: We only measure the False Positive rate of first-time arrivals of a message, making the False Positive rate equal to $\sum_i F_i / \sum_i 1$

\end{enumerate}

Which variant to use depends on the needs of the underlying application.  For instance, a False Positive in a caching application that uses the BF to determine cache hits for requests would make the wrong assessment for an initial arrival of a request, but the mistake would be detected upon attempting to fetch the item from the cache, so the item would be cached and subsequent requests would not be false positives.  Hence, count-first is the appropriate measure.  Alternatively, a mechanism that wishes to count unique messages would have a single miscount for each message initially classified incorrectly as having been previously received, hence count-instance is the appropriate measure.  Finally, in a setting where the BF is used to drop repeat requests, and where a penalty is paid per dropping of an unsatisfied request, count-each would be the appropriate measure.

We can prove that count-instance is the most stringent (largest) of False Positive rates of the three versions\footnote{count-first is straightforwardly smaller count-each.}.  The proof showing that count-instance results in the largest False Positive rate is omitted here but is included in \S\ref{sec:count-instance-proof} in the Appendix.  This observation yields two benefits.  First, if we compute the false positive rate for count-instance, we have upper bounds on the rates of the count-first and count-each variants.\footnote{In fact, our analysis of count-instance is itself an upper-bound in an RBF for a subtle reason that we address in the extended version.}  Second, count-instance is not affected by underlying distribution of arriving messages: how often a particular message repeats is of no import, since we only decide its contribution to False Positive rate based on its first arrival, and any subsequent arrivals do not alter the bits in the BF, making its analysis identical regardless of this underlying distribution, whereas the other  variants do depend on the underlying distribution.

The remainder of this paper utilizes count-instance, i.e., the False Positive rate is defined with respect to how new arrivals are classified upon their arrival, as it is the most stringent of the three, and provides an upper bound on the other interpretations of false-positive rate.

\section{\texorpdfstring{$N$}{N}-bounded average False Positive rate}
\label{sec:n-bounded}

We compute a lower bound on the $N$-bounded average False Positive rate that utilizes the formula which determines worst-case False Positive rate (i.e., (\ref{eq:wikiform}) for colliding hash functions).  Define $F_i$ to be an indicator r.v. that equals 1 when the $i$th new arrival is a False Positive.  Note that $E[F_i] = P(F_i=1)$ is solved directly by worst-case bound formula (\ref{eq:wikiform}) directly.

First, let us consider an {\em oracle} user, who, when inserting a message that sets no bits, can distinguish whether this message is a repeat, or a new message that failed to set bits.\footnote{This oracle is a hypothetical user because the user can make such a distinction has no need for a BF.} This oracle would insert exactly $N$ new messages, and the resulting average False Positive rate can be computed as a simple renewal process, yielding a count-instance {\em oracle $N$-bounded average} False Positive rate of:
\begin{IEEEeqnarray}{cLr}
\label{eq:oraclebound}
f_o & = & \sum_{i=1}^N E[F_i] / N.  
\end{IEEEeqnarray}

In contrast, a ``real'' user would assume that a new message that sets no bits is instead a repeat message, and would not include such a message in their count of $N$ messages.  Hence, the oracle would recycle no later than an actual user, thereby ensuring a lower false positive rate. 
 For this real user, we let $i$ iterate over the number of new messages that set bits, and define r.v. $R_i$ to equal the number of new messages that  do not set bits and arrive between the $i-1$st and $i$th new messages that do set bits.  Then for a single cycle of the RBF, the False Positive rate is $\sum_{i=1}^N R_i / \sum_{i=1}^N (R_i + 1)$.

Note that for the messages that comprise a given $R_i$, the number of bits set in the BF remains constant (none of the $R_i$ messages are changing the number of bits set), so each such message has an equal likelihood, which we call $p_i$, of not setting additional bits.  Note it follows that $P(R_i > j) = {p_i}^{j+1}$, such that $E[R_i] =  p_i (1-p_i)^{-1}$.  By renewal theory, we have the False Positive rate is thus equal to $$\frac{\sum_{i=1}^N p_i (1-p_i)^{-1}}{\sum_{i=1}^N (1+p_i (1-p_i)^{-1})}.$$

Finally, we show that $p_i$ is lower-bounded by $P(F_i=1)$.  This follows from the fact that the $i$th new message that set bits is actually the $j$th arriving message with $j \ge i$, so $p_i = P(F_j=1)$ for some $j \ge i$.  The lower bound then follows from the fact that $P(F_i=1)$ is increasing in $i$ (False Positive rates increase as more messages arrive), and that the worst-case analyses' value for $N$ is over all messages, including those that set no bits.  Hence, $p_i = P(F_j=1) \ge P(F_i=1)$ for $j \ge i$.  Since we are underestimating each $p_i$, we are underestimating the likelihood of new arrivals not setting bits, so we are underestimating False Positive rates.  Note that for low False Positive rates that are used in practice, $E[R_i]$ will tend to be very small, such that the bound is expected to be tight.

It will be useful for our purposes to define, for $f_i = P(F_j = 1)$ and directly applying \eqref{eq:wikiform},  a lower bound on the  $N$-bounded False Positive rate of \begin{equation}
\label{eq:avgN}
  f_a = \frac{\sum_{i=1}^N f_i (1-f_i)^{-1}}{\sum_{i=1}^N (1+f_i (1-f_i)^{-1})}.  
\end{equation}

We refer to this as the {\em average-case N} lower bound. Note that these lower-bound the False Positive rate as a function of $N$.  In \S\ref{sec:sims}, we use these formulae to determine $N$ as a function of False Positive rate: since the formula gives a lower bound on the False Positive rate for a given $N$, it indicates an upper bound on the value of $N$ needed to ensure (lower-bound) a given False Positive rate.

\section{\texorpdfstring{$\sigma$}{sigma}-bounded average False Positive rate}
\label{sec:sigma-bounded}

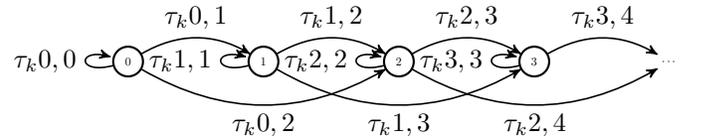
\begin{figure}[!htb]
    
\begin{center}
    
    \begin{tikzpicture}[->, >=stealth', auto, semithick, node distance=4cm]
	\tikzstyle{every state}=[fill=white,draw=black,thick,text=black,scale=0.45]
    \tikzstyle{dots} = [draw=none, scale=0.45]
	\node[state]    (A)                     {$0$};
	\node[state]    (B)[right of=A]   {$1$};
	\node[state]    (C)[right of=B]   {$2$};
	\node[state]    (D)[right of=C]   {$3$};
    \node[dots] (E) [right of=D] {$\cdots$};
	\path
	(A) edge[loop left]			node{$\tau_{k}{0,0}$}	(A)
    edge [bend left, above] node{$\tau_{k}{0,1}$} (B)
    edge [bend right, below] node{$\tau_{k}{0,2}$} (C)
    (B) edge [loop left] node{$\tau_{k}{1,1}$} (B)
    (C) edge [loop left] node{$\tau_{k}{2,2}$} (C)
	(B) edge[bend left,above] node{$\tau_{k}{1,2}$} (C)
    edge [bend right, below] node{$\tau_{k}{1,3}$} (D)
	(C) edge[bend left,above]	node{$\tau_{k}{2,3}$}	(D)
    edge [bend right,  below] node{$\tau_{k}{2,4}$} (E) 
	(D) edge[loop left] node{$\tau_{k}{3,3}$}	(D)
 (D) edge[bend left, above] node{$\tau_{k}{3,4}$} (E);
 
	\end{tikzpicture}
 \centering
    \caption{Partial Markov Chain diagram for Recycling Bloom Filter, showing the ``forwards'' transition behavior, for $k=2$ using colliding hash functions.  States represent the number of RBF bits set to $1$ and $\tau_{k}{i,j}$ is the probability of transitioning from state $i$ to state $j$.}
    \label{fig:Markov1}
 \end{center}
    
\end{figure}

As a starting point towards developing a recursive analytical model that accurately captures the average False Positive rates for the $\sigma$-bounded model, we consider a process where messages sequentially arrive at the RBF.  Once a message arrives, the $k$ hash functions are applied to it and the corresponding bits are set to $1$ (or remain $1$ if already set).  In a departure from the $N$-bounded method of Bloom Filter recycling, we do not reset the RBF after $n$ unique message insertions.  Instead, we have a predetermined bit capacity parameter $\sigma < M$.  After the insertion of a message, let $b$ be the number of bits set to 1 in the RBF.  We reset the RBF when $b > \sigma$.  Our recursive Markov Model will allow us to determine the long-term average False Positive rate across all new arrivals for a given value of $\sigma$ (another departure from the traditional False Positive analysis, which only considers the ``worst case'' False Positive rate for the last arrival of a RBF reset cycle).

\subsection{RBF Recursive Markov Model}
 Our Markov Model for the dynamics of the RBF has $\sigma + 1$ states labeled $0, 1, \cdots, \sigma$.  The state labels correspond to the number of bits set to $1$ in the RBF during a given cycle.  After the arrival of a {\em unique} message (i.e., not yet seen this cycle), depending on the outcome of the hash operation, the number of set bits in the RBF may or may not change.  We represent this event by transition probabilities between the states, labeled  $\tau_{k}{i,j}$.  The number of employed hash functions, $k$, determines which transition probabilities are nonzero-- e.g. from state $i$, the ``largest'' possible state we can transition to is state $\text{min}(i + k, \sigma)$.  A partial diagram of our Markov Model, displaying only the ``forward'' transition behavior, is shown in Fig. ~\ref{fig:Markov1}. 

 \subsection{Retaining vs. Non-Retaining RBF}

 We model the resetting of the RBF with ``backwards'' transition probabilities from certain states $j$ to earlier states $i < j$.  In practice, RBFs can implement recycling behavior in two ways: if message $m$ causes the RBF to reset ($b > \sigma)$, after clearing all the bits, we can:

 \begin{enumerate}[label=\roman*]
     \item \label{firstchoice} Retain message $m$ and re-insert it as the first message for the new RBF cycle
     \item \label{secondchoice} Forget message $m$ and have the {\em next} arrival as the first message for the new RBF cycle
 \end{enumerate}

 We call (\ref{firstchoice}) the {\em retaining} version of the RBF and (\ref{secondchoice}) the {\em nonretaining} version.  For both the retaining and non-retaining RBF, nonzero backwards transition probabilities exist at states $\{\sigma - k + 1, \sigma - k + 2, ... , \sigma\}$. For the retaining RBF, the backwards transitions are to states $\{1, 2, ... , k\}$.  For the non-retaining RBF, all backwards transitions are to state $0$.   While we analyze both versions, the non-retaining version is slightly simpler for analysis purposes and is the version we work with unless otherwise mentioned. Fig. ~\ref{fig:Markov2} shows a partial diagram of our Markov Model displaying the backward transitions.

  
    
    
 
   
   

\begin{figure}[!hb]
    \centering
    \includegraphics[width = \columnwidth]{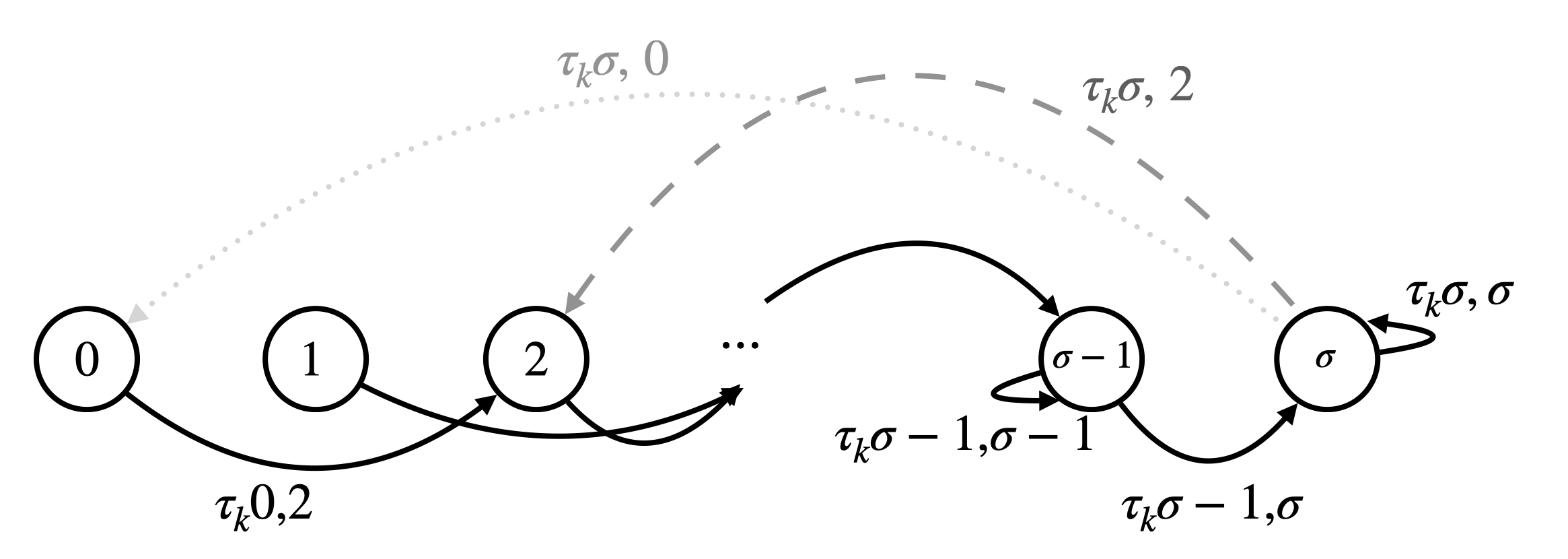}
    \caption{Partial Markov Chain diagram for Recycling Bloom Filter showing the ``backwards'' transition behavior, for $k=2$ (some transition arrows omitted). The light grey arrow depicts the transition for the {\em retaining version}, while the dark grey arrow is for the {\em non-retaining version}, non-colliding.}
    \label{fig:Markov2}
\end{figure}
 \subsection{RBF Long-Term Average False Positive Rate}

With our Markov Model specified, we can derive an expression for the long-term average False Positive Rate of the RBF, for both colliding and non-colliding BFs, and also derive the backward transitions for colliding/noncolliding and retaining/non-retaining versions (four in total).

\subsubsection{Transition Probabilities}

We start by deriving expressions for the forward transition probabilities $\tau_{k}(i,j)$ of an arriving (new message), where $k$ is the number of hash functions, $i$ is the number of bits set to 1 in the BF prior to insertion of the arriving message, and $j$ is the ending number of bits set to 1 (after the next message is hashed into the filter).   We accomplish this for forward transitions from $i$ to $j \ge i$  by recursively defining $\tau_{k}(i,j)$ in terms of $\tau_{k-1}(i,j)$ and $\tau_{k-1}(i,j-1)$. The explicit recursive equation depends upon whether the BF is being implemented using colliding (vs. non-colliding) hash functions, as well as whether it is retaining (vs. non-retaining).  The intuition behind the recursion is that the $k$th hash function will add at most one bit in the BF after the other $k-1$ hashes have been applied.  The BF transitions from $i$ to $j$ bits being set if either a) the first $k-1$ hashes transition to $j$ bits and the last hash is to a bit previously set (by earlier messages and/or the preceding $k-1$ hashes of the current message), or b) the first $k-1$ hashes transition to $j-1$ bits set, and the last hash maps to an unset bit. 
Note that these recursive equations make use of the independence properties of the hash functions applied within Bloom Filters; each message's hashes are independent of those of previous hashes, colliding hashes for a given message are also independent from one another, and non-colliding hashes are chosen via sampling without replacement.  The former property ensures that a message arriving to a BF with $b$ bits set can assume these $b$ bits were simply sampled uniformly at random.

We define  $\cn{k}(i,j)$ to be the transition probability when the BF uses colliding hash functions and is non-retaining.  In this case when $j \ge i$:
\begin{equation}
\label{eq:nonret}
    \cn{k}(i,j) = \cn{k-1}(i,j) \frac{j}{M} + \\ \cn{k-1}(i, j-1)\frac{M-j+1}{M}
\end{equation}
\begin{equation}
    \cn{k}(i, 0) = \cn{k-1}(i,0) + \cn{k-1}(i, \sigma)\frac{M-\sigma}{M}
\end{equation}

And the base case relations:
\begin{eqnarray}
    \label{eq:base8}
    \cn{0}(i,j) & = & I_{i=j}\\
    \cn{k}(0,0) & = 
    & 0; k>0 \\
    \label{eq:base3}
    \cn{k}(i,j) & = & 0; \phantom{-} j > i + k, k>0 
\end{eqnarray}
where $I_X$ is an indicator that equals 1 when $X$ is true, and is otherwise 0.  For colliding + retaining, denote $\tau_{k}(i,j)$ as $\crr{k}(i,j)$.  $\crr{k}(i,j)$ has identical form to equation (\ref{eq:nonret}) for $j \ge i$.  We also have:
\begin{equation}
    \crr{k}(i,0) = 0
\end{equation}
\begin{equation}
    \crr{k}(i,j) = \cn{k}(i,0) \cdot \cn{k}(0,j); \phantom{-} j \le k < i 
\end{equation}

The base cases for $\crr{0}(i,j)$ have identical form as (\ref{eq:base8})-(\ref{eq:base3}).  
For non-colliding + non-retaining, denote $\tau_{k}(i,j)$ as $\nn{k}(i,j)$. We have the recursive relations:  
\begin{multline}    
    \nn{k}(i,j) = \nn{k-1}(i,j) \frac{j-(k-1)}{M-(k-1)} +  \\ 
    \nn{k-1}(i,j-1)\frac{M-j+1}{M-(k-1)}; \phantom{-} j \ge i
\end{multline}
\begin{equation}
    \nn{k}(i,0) = \nn{k-1}(i,0) + \nn{k-1}(i, \sigma)\frac{M-\sigma}{M-(k-1)}
\end{equation}
\begin{equation}
    \nn{k}(0,k) = 1
\end{equation}
The base cases for $\nn{0}(i,j)$ have identical form as (\ref{eq:base8})-(\ref{eq:base3}).

For non-colliding + retaining, denote $\tau_{k}(i,j)$ as $\nr{k}(i,j)$.  $\nr{k}(i,j)$ has identical form to $\nn{k}(i,j)$ for $j \ge i$.  We also have:
\begin{equation}
    \nr{k}(i,k) = \nn{k}(i,0); \phantom{-} i > k
\end{equation}
The base cases for $\nr{0}(i,j)$ have identical form as (\ref{eq:base8})-(\ref{eq:base3}). 

\subsubsection{Steady-state Probabilities}

For our Markov Model, let us denote by $\pi_{i}$ the steady-state probability the RBF has $i$ bits set, where $0 \le i \le \sigma$ (this is equivalent to the probability of the Markov chain being in state $i$).  

For both the colliding + non-retaining / non-colliding + non-retaining versions of the RBF, the Markov Chain is clearly ergodic. Therefore the stationary distribution $\pi_{i}$ can be computed by solving for $\pi_i$: 
\begin{equation}
    \pi_{i} =\frac{ \sum_{j=\text{max}(0,i-k)}^{i-1} \pi_{j} \tau_{k}(j,i) }{1 - \tau_{k}(i,i)}
\end{equation}
We can solve the equations for all $i$ by starting with an arbitrary positive assignment (a ``guess'') for $\pi_{0}$; call this $g_{0}$. Then, we apply (12) to recursively compute $g_{i}$ for $\pi_{i}, i > 0$ as functions of $g_{0}$.  Since the sum over all steady state probabilities must be 1, we can renormalize each $\pi_{i} = \frac{g_{i}}{\sum_{j=0}^{\sigma}g_{j}}$.

For the steady-state probabilities of the non-colliding, non-retaining version of the RBF, we can apply an identical method as above, except we have $\pi_{i} = 0$ for $i < k$.  Thus, we make our initial ``guess'' for the steady-state probability $\pi_{k}$, but the sets of equations are otherwise unaltered from the previous case.

For the steady-state probabilities of the colliding, retaining version of the RBF, we have an additional complexity.  Equation (12) now holds only for $i > k$.  For $0 \le i \le k$, due to the backwards transitions to these states, the steady-state probabilities are given by:
\begin{equation}
\label{linearPi}
    \pi_{i} = \frac{\sum_{j=1}^{i-1}\pi_{j}\tau_{k}(j,i) + \sum_{j = \sigma - k + 1}^{\sigma} \pi_{j}\tau_{k}(j,i)}{1 - \tau_{k}(i,i)}
\end{equation}
To solve for the set of $\pi_{i}$, we can take the set of equations of (12) and (13), along with the constraint $\sum_{j=0}^{\sigma}\pi_{i} = 1$, and solve by standard linear methods.

\subsubsection{Expression for Long-Term Average False Positive Rate}
Let us denote by $\text{$\rho$}(i)$ the long-term average False Positive rate when in state $i$.  

For both the colliding and non-colliding versions of the RBF, we have:
\begin{equation}
\label{eq:rho-def}
    \rho(i) = \left\{ \begin{array}{ll}
    \frac{\binom{i}{k}}{\binom{M}{k}} & \mbox{non-colliding}\\
    (\frac{i}{M})^{k} & \mbox{colliding}\\
    \end{array} \right \}
\end{equation}
Then the overall long-term average False Positive rate for the RBF is given by (using the desired definition of  $\rho(i)$ from (\ref{eq:rho-def}):
\begin{equation}
\label{longTermFalsePosRate}
    f_{\sigma}^{1} = \sum_{i=0}^{\sigma} \pi_{i} \cdot \rho(i)
\end{equation}

\subsection{Closed-Form expression for $k=1$}

For the case of colliding, non-retaining, and $k=1$, we can explicitly solve the above expressions in closed form for the $\pi_{i}$, yielding a closed form expression for the false long-term average False Positive rate:

\begin{equation}
\label{longTermFalsePosRateExact}
    f_{\sigma}^{1,k=1} = \sum_{i=0}^{\sigma} \frac{i}{M(M-i)\sum_{j=0}^{\sigma}\frac{1}{M-j}}
\end{equation}

\subsection{Computational Complexity of Calculations}

Our computation proceeds with an outer-most loop iterating over $k$, with $\tau_0(i,j) = 1$ only when $i=j$ and is otherwise 0.  For the next value of $k$, we compute $\tau_k(i,j)$ for values of $i$ and $j$ ranging between 0 and $\sigma+k$.  Noting that $\tau_k(i,j)=0$ for $j>i+k$ (hashing to $k$ bins cannot set more than $k$ bits), this involves $(\sigma+k)k$ total computations, making the overall time complexity of $O(Mk^2)$.  If only one value of $k$ is being evaluated, it is possible to produce the table of all necessary $\tau_k(i,j)$ holding only $\sigma (k+1)$ values, so it can also be done efficiently in memory.

Once the requisite $\tau_{k}(i,j)$ have all been computed, computing the steady-state probabilities $\pi_i$ for the {\em non-retaining} case can be done in $O(kM)$ time, since there are a total of $\sigma+1$ steady-state probabilities ($\sigma \le M$) and each of them involve a sum across $k$ terms.

Computing the steady-state probabilities $\pi_i$ for the {\em retaining} case can be done in $O(k^{2}M)$ time.  There are a total of $\sigma+1$ steady-state probabilities ($\sigma \le M$), and $\sigma + 1$ corresponding equations for $\pi_i$.  Each of these equations contain $k+1$ terms; linearly combining a pair of rows is therefore an $O(k)$ operation.  Thus, with $O(M)$ row reductions, each with complexity $O(k)$,  we can compute $\pi_i$ for all $i \le k$ (total complexity $O(k^{2}M)$.  We can subsequently compute $\pi_i$ for $ i > k$ directly from \eqref{linearPi}; this takes $O(kM)$ time, so the row reduction step is the complexity bottleneck. 

Computing the long-term average RBF False Positive rate as in \eqref{longTermFalsePosRate} can be done in $O(kM)$ time for colliding and $O(M)$ time for non-colliding (using the fact that $\binom{i+1}{k} = \binom{i}{k}\frac{i+1}{i+1-k}$ to compute the $\rho(i)$ terms in constant time).

Since the long-term average RBF False Positive rate is computed by successive and separate (but dependent) computations each upper-bounded by $O(k^{2}M)$, the overall computational complexity is $O(k^{2}M)$.

\subsection{Expected messages within a \texorpdfstring{$\sigma$}{sigma}-bounded RBF}
\label{sec:expectedN}
In each cycle of a $\sigma$-bounded RBF, the number of messages that are hashed into the BF prior to recycling will vary, depending on the number of hash collisions between (and for colliding, also within) messages.  We conclude our analysis of the one-phase BF by showing how the expected number of messages can be computed.

Consider a particular sample path (i.e., cycle) of the RBF, and let $b(i)$ indicate the number of bits set after the arrival of the $i$th new message.  Let $N(b(i))$ equal the number of additional messages sent after the $i$th message to trigger a recycle (i.e., cross the sigma threshold).  Note that when $b(i) \ge \sigma$, we have already crossed the threshold such that $N(b(i)) = 0$.  Otherwise, when $b(i) < \sigma$, more messages must be received, such that
\begin{IEEEeqnarray*}{lCr}
N(b(i)) = 1 + N(b(i+1)) & = & 1 + \sum_{j=0}^k X_{j}(i) N(b(i)+j)
\end{IEEEeqnarray*}
where $X_j(i)$ is an indicator that equals 1 only when $j$ additional bits get set from the arrival of the $i$th message.

Noting the above equation holds irrespective of the value of $i$, we can simply replace $b(i)$ with $b$ and just write $N(b)$, which we replace with $N_b$ to get the result $N_b = 1 + \sum_{j=0}^k X_{j}(i) N_{b+j}$.  We can similarly substitute in the r.v. $X_{b,b+j}$ for $X_{j}(i)$ which indicates that this $i+1$st message takes the BF from having $b$ bits set to $b+j$. This can be solved to permit a reverse recursion:
\begin{IEEEeqnarray}{lCr}
    N_b =&  0, & b \ge \sigma\\
    N_b =&  \frac{1 + \sum_{j=1}^k X(b,b+j) N_{b+j}}{1 - X(b,b)},& b < \sigma
\end{IEEEeqnarray}
 Noting independence of the $X_{b,b+j}$ from $N_b$, and that $E[X_{b,b+j}] = \tau_k(b,b+j)$, we can rephrase the above as an expectation:
\begin{IEEEeqnarray}{lCr}
    E[N_b]   =&  0, & b \ge \sigma\\
    E[N_b]  =& \frac{1 + \sum_{j=1}^k \tau_k(b,b+j) E[N_{b+j}]}{1 - \tau_k(b,b)},&  b < \sigma
\end{IEEEeqnarray}
and our solution is simply $E[N_0]$.

\subsubsection{Two-Phase RBF Long-term average False Positive Rate}

\label{sec:twophase}

We can use the results from the standard RBF to also derive an expression for the False Positive rate of the two-phase RBF variant.
Let $F_{i}$ denote the steady-state probability of the {\em frozen} filter having $i$ bits. Note $F_{i}$  is only nonzero for  $i$ between $\sigma-k+1$ and $\sigma$. To compute $F_{i}$, we first define an un-conditional (non-normalized) value of $\hat{F}_{i}$ to then compute a normalized (conditional) version of $F_i$ in terms of $\hat{F}_i$:
 
\begin{IEEEeqnarray}{lCr}
    \hat{F}_{i} & = & 
    \pi_{i} \sum_{j=\sigma+1}^{i + k}\tau_{k}(i,j)\\
    F_{i} & = & \frac{\hat{F}_{i}}{\sum_{j=\sigma-k+1}^{\sigma}\hat{F}_{i}}
\end{IEEEeqnarray}

The long-term average False Positive rate for the two-phase RBF is then given by the expression (using the desired definition of  $\rho(i)$ from (\ref{eq:rho-def}):

\begin{equation}
\label{eq:twophase}
    f_{\sigma}^{2} = 1 - (1 - \sum_{j=0}^{\sigma}\pi_{j}\rho(j))(1 - \sum_{i = \sigma-k+1}^{\sigma}F_{i} \rho(i))
\end{equation}

\section{Results}
To evaluate our models and draw conclusions about the performance of the RBF under varying parameters, we proceed in three steps:  
\begin{enumerate}
    \item We verify the accuracy of our models through discrete event-driven simulations implemented in Python.  The RBF data structures were implemented using standard Python libraries. The uniform hash function generation and random message arrival process leveraged the Python \texttt{random} library.  Code can be found at: \cite{SimCodeRepository}

    \item With the accuracy of the models verified, we turn to the question of the maximum message capacity achievable by a RBF while staying below a given False Positive rate.  We consider four models of a one-phase RBF: the $\sigma$-bounded model (\ref{longTermFalsePosRate}), the ``oracle'' $N$-bounded False Positive rate \eqref{eq:oraclebound}, the lower-bound on the $N$-bounded average False Positive rate, $f_a$ \eqref{eq:avgN}, and the traditional worst-case bound  $f_w$ \eqref{eq:wikiform}. The $\sigma$-bounded model maintains the highest number of messages, followed by $f_o$, $f_a$, and $f_w$ the lowest. This sequence implies that $N$-bounding variants are overly conservative estimates. Therefore, if one aims to size a RBF to achieve the best performance, the $\sigma$-bounded variant should be utilized.

    \item Finally, for the $\sigma$-bounded model, we investigate the trade-off between using a one-phase and two-phase RBF variant, in terms of the maximum memory capacity achievable by each variant while staying below a given False Positive rate. 
\end{enumerate}
     
We note that the results for the combinations of colliding/non-colliding and retaining/non-retaining variants are nearly identical, and thus for brevity all graphs represent the colliding/non-retaining combination.

\label{sec:sims}

\subsection{Verification of Model Accuracy}
\label{verification_of_model}
\subsubsection{One and Two-Phase Sigma Bound}

We verify the accuracy of $f_{\sigma}^{1}$ and $f_{\sigma}^{2}$ through  discrete event-driven simulations.  For all simulations, we have a fixed RBF size of $M = 1000$ bits.  The two-phase RBF splits this memory evenly between each filter.  We run multiple simulation epochs of $100,000$ message arrivals, recording the average False Positive rate at the end of each epoch.  Message arrivals are drawn uniformly from a pre-generated distribution of $1000$ messages.  Given a collection of average False Positive rates, we can plot sample means and confidence intervals using standard statistical methods \cite{sidik2002simple}.  We compare these results to the values predicted by $f_{\sigma}^{1}$ and $f_{\sigma}^{2}$.  Fig. \ref{fig:onePhaseVerifyAcc} depicts the results.  We simulated for a total of $7$ epochs for $f_{\sigma}^{1}$ and $10$ epochs for $f_{\sigma}^{2}$. The $99$\% confidence interval is indicated on the plot by the shaded area around the simulation curves.  For all data points, the value of $f_{\sigma}^{1}$ and $f_{\sigma}^{2}$ lies within the confidence interval.  This is strong evidence in favor of the accuracy of $f_{\sigma}^{1}$ and $f_{\sigma}^{2}$  in modelling False Positive rates for RBFs.

\begin{figure}[!htb]
\centerline{\includegraphics[width=\columnwidth]{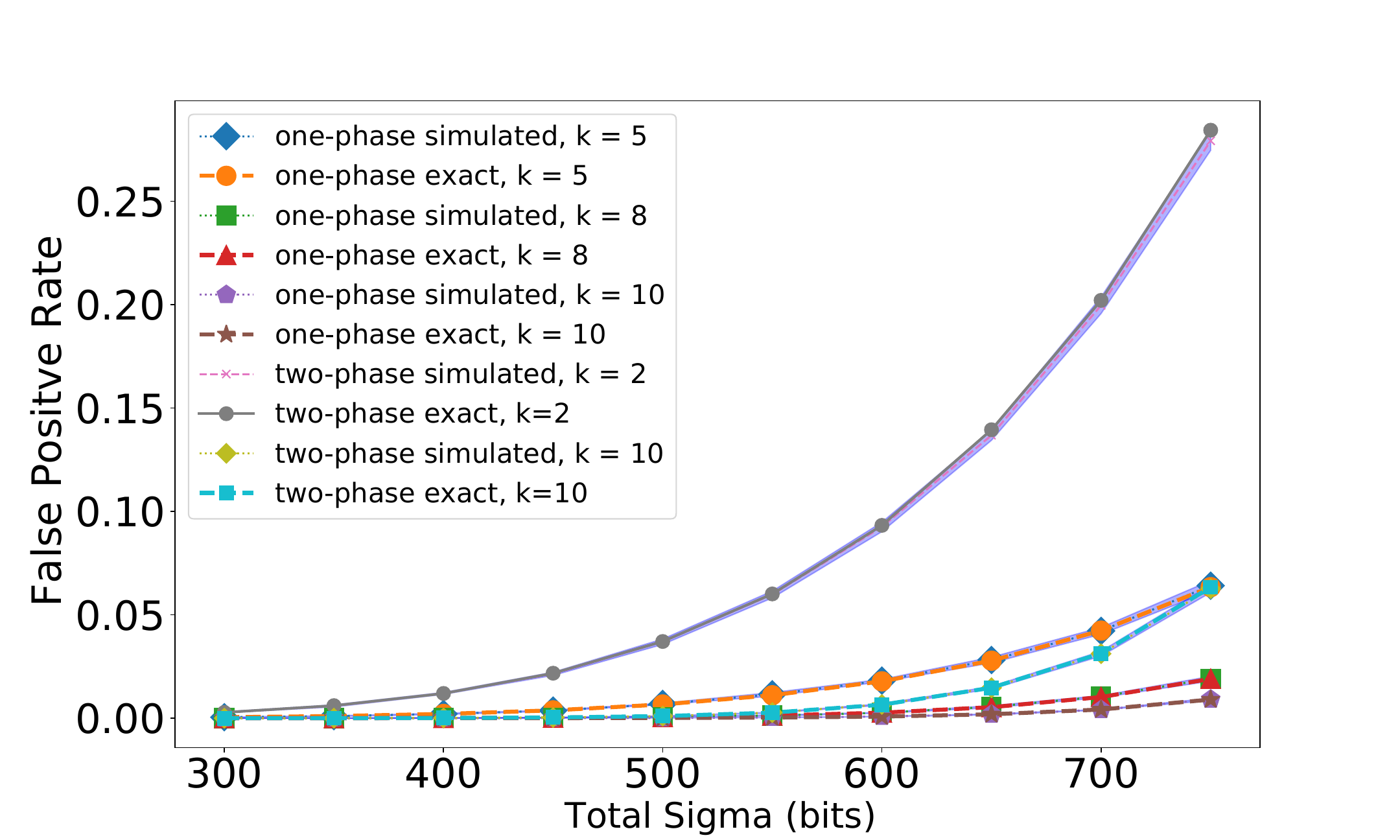}}
\caption{Verifying our model accuracy through simulation.}
\label{fig:onePhaseVerifyAcc}
\end{figure}

\begin{figure}[htbp!]
\centerline{\includegraphics[width=\columnwidth]{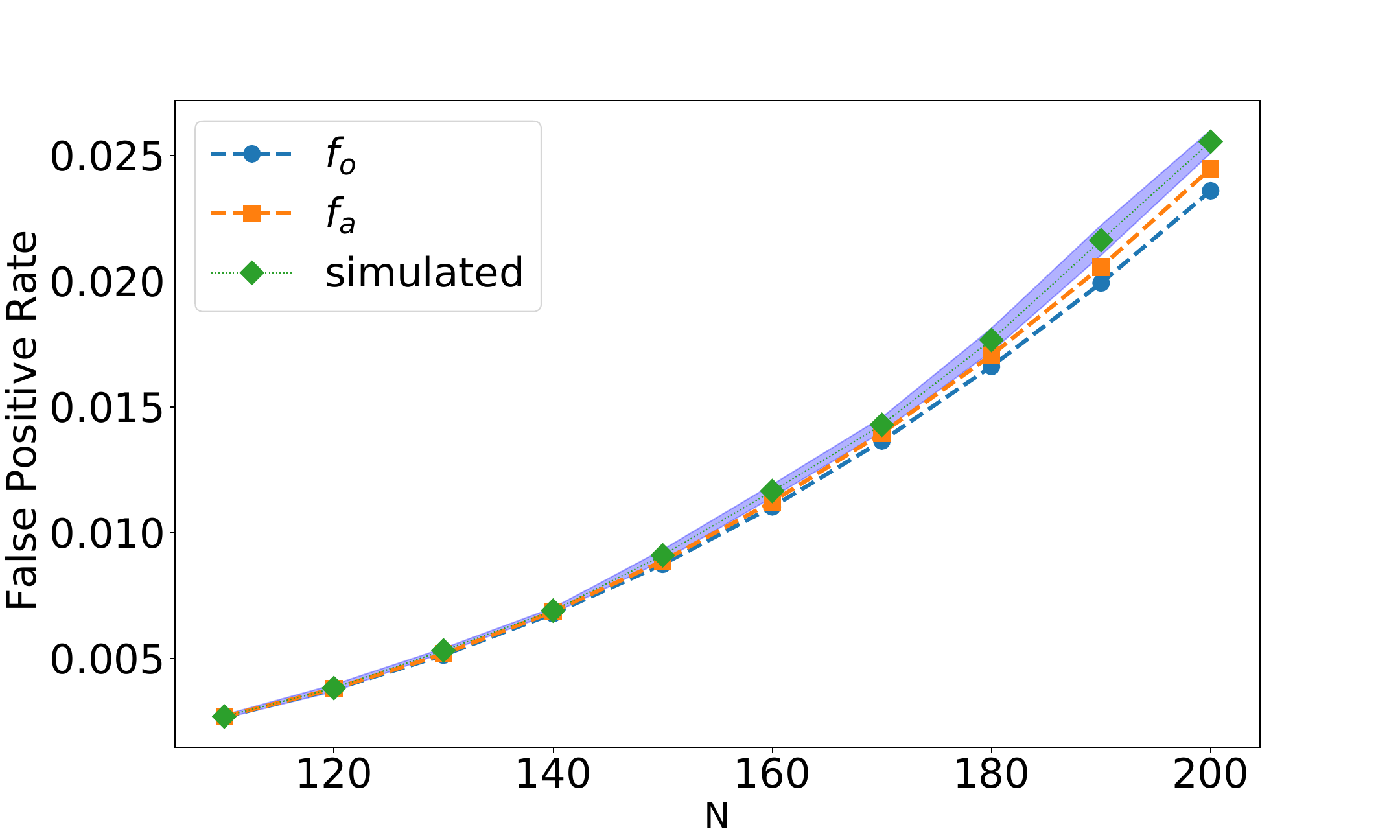}}
\caption{Verifying $f_o, f_a$ are lower bounds for average False Positive rate through simulation.}

\label{fig:avereageNLowerBound}
\end{figure}

 
\subsubsection{Lower bounds on False Positive rate}

We turn next to verifying the accuracy of $f_{a}$ and $f_{o}$.  For this, we ran $14$ simulation epochs of $1,000,000$ messages each.  Fig. \ref{fig:avereageNLowerBound} shows the simulation results plotted against $f_{a}$ and $f_{o}$, with the $99$\% confidence interval shown.  For all data points, both $f_{o}$ and $f_{a}$ are below the sample mean, and fall either within or below the confidence interval.  This demonstrates strong evidence in favor of the accuracy of the lower bounds $f_{o}, f_{a}$.  Observe that $f_{a}$ is a ``tighter'' lower lower bound than $f_{o}$.

\subsection{Expected Message Capacity Comparison}
\label{expected_number}

A natural question to ask when optimizing RBF parameters is the expected maximum number of messages one can store while staying under an acceptable average False Positive rate-- we refer to this as the {\em expected message capacity}.  For the one-phase $\sigma$-bounded model (False Positive rate given by $f_{\sigma}^{1}$),  recall this is $E[N_0]$ derived in \S\ref{sec:expectedN}.  We can also compute an expected message capacity using $f_{a}$ or $f_{\omega}$.   Fig.   \ref{fig:compareVaryM} compares the expected message capacity between the three models.  For each model, we fix an average False Positive rate of $.01$. For each value of $M$, we find the maximum number of messages subject to the False Positive constraint.  We plot this number normalized to $f_{\sigma}^{1}$, as this is the model that yields the highest message capacity.  $f_{a}$ is seen to be a good approximation,  a direct consequence of the tightness of the lower bounds discussed in \S\ref{sec:n-bounded}.  Observe that $f_{\omega}$ can lead to an overly conservative estimate of message capacity in this context; for the given parameters its maximum message capacity is  consistently reduced by more than 30\%. 

\begin{figure}[htbp!]
\centerline{\includegraphics[width=\columnwidth]{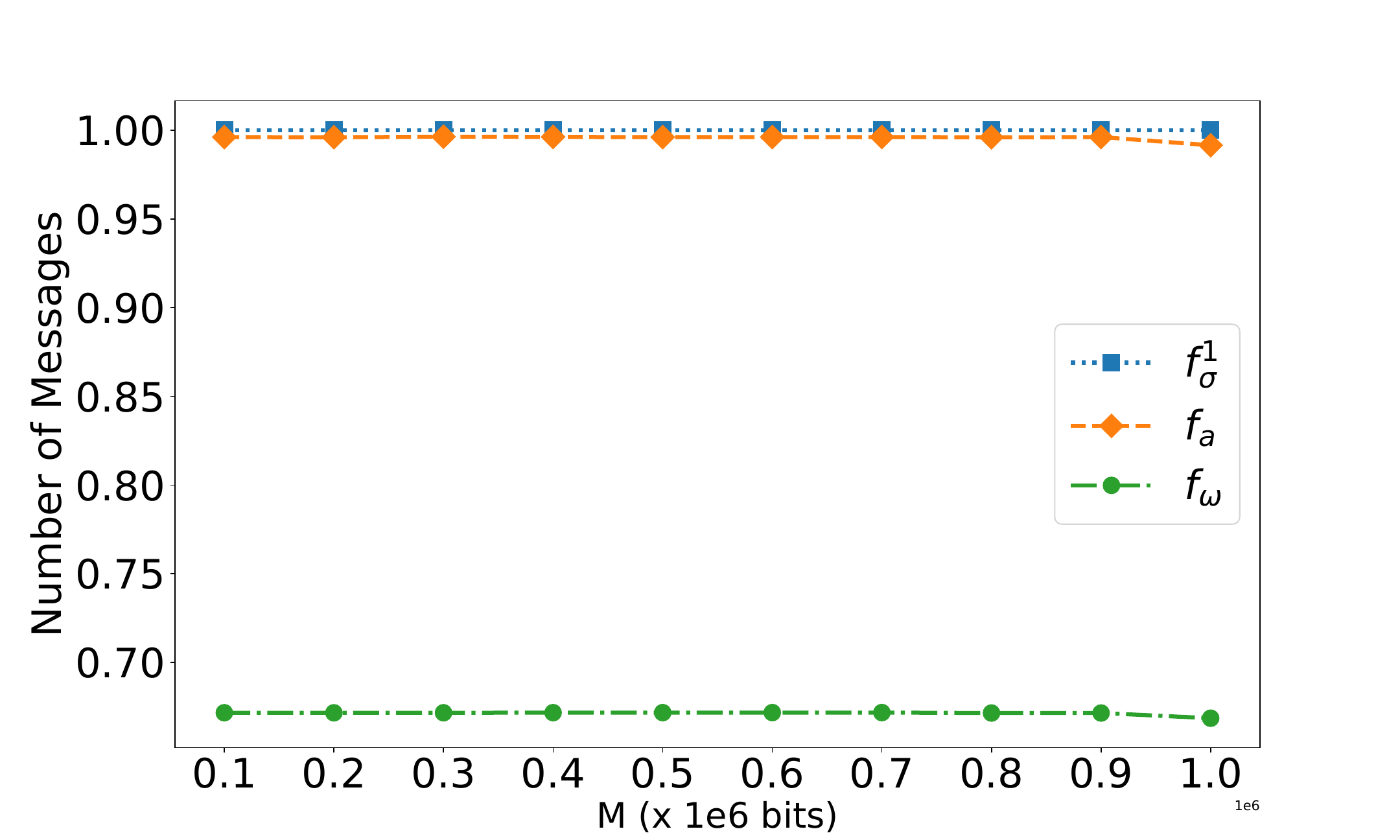}}
\caption{Expected message capacity of $f_w, f_a$ normalized by $f_{\sigma}^{1}$ message capacity, for different values of $M$.}

\label{fig:compareVaryM}
\end{figure}

\begin{figure}[htbp!]
\centerline{\includegraphics[width=\columnwidth]{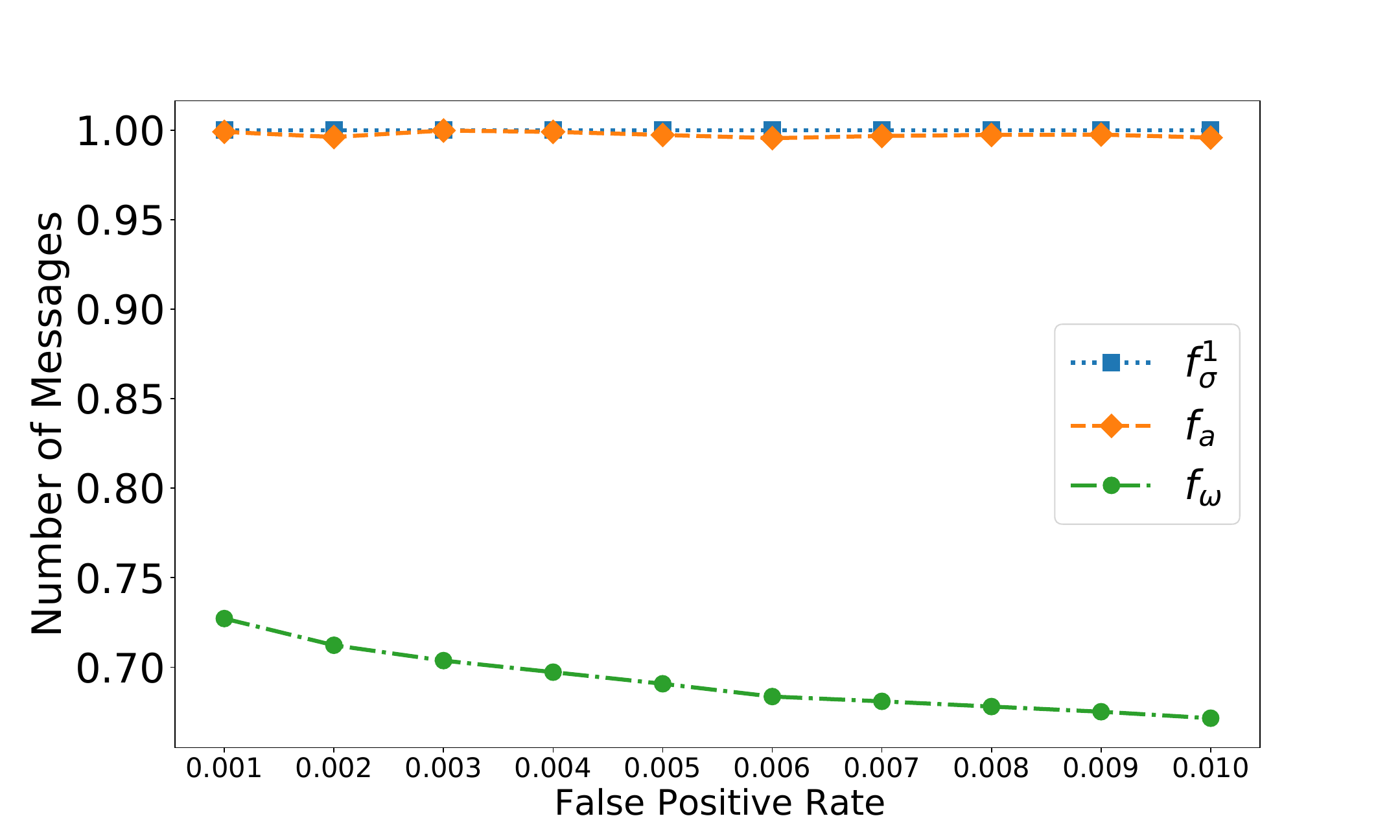}}
\caption{Expected message capacity of $f_w, f_a$ normalized by $f_{\sigma}^{1}$ message capacity, for different False Positive rates.}
\label{fig:compareVarySigma}
\end{figure}

Fig. \ref{fig:compareVarySigma}. shows similar results, this time varying False Positive rates on the x-axis.  Once again, the worst-case $N$ model sees a consistent underestimation of maximum message capacity of around $30$\%.  

We can easily conclude that, for applications that are more concerned with average False Positive rates than worst-case ones, $f_{\omega}$ not an ideal metric, as compared to $f_{\sigma}^{1}$ or $f_{a}$.



\subsection{Tradeoff between One and Two-phased RBF}

In Fig. \ref{fig:oneTwoCompM}, we compare the False Positive rates of one-phase RBFs ($f_{\sigma}^{1}$) versus their two-phase counterparts ($f_{\sigma}^{2}$). 
  We have a given filter memory $M$ (which in the case of the two-phase RBF is split evenly between both filters) and False Positive rate.  We then find the value of $k$  which allows us the maximum number of expected messages for these constraints. In both cases, we observe the result that the one-phase RBF appears to outperform the two-phased filter when it comes to squeezing out extra filter bit capacity.  One might be moved to ponder what is the point of the added complexity of a two-phased filter, if it apparently is ``less efficient'' in terms of average False Positive rate?  The answer becomes apparent when we consider the idea of {\em False Negatives}  as we previously mentioned.  While a detailed investigation of False Negatives is beyond the scope of this paper, the intuition is that the extra bit capacity seemingly afforded by the one-phase filter comes with an increased probability of False Negative rates, which are mitigated by the two-phase filter.

\begin{figure}[htbp]
\centerline{\includegraphics[width=\columnwidth]{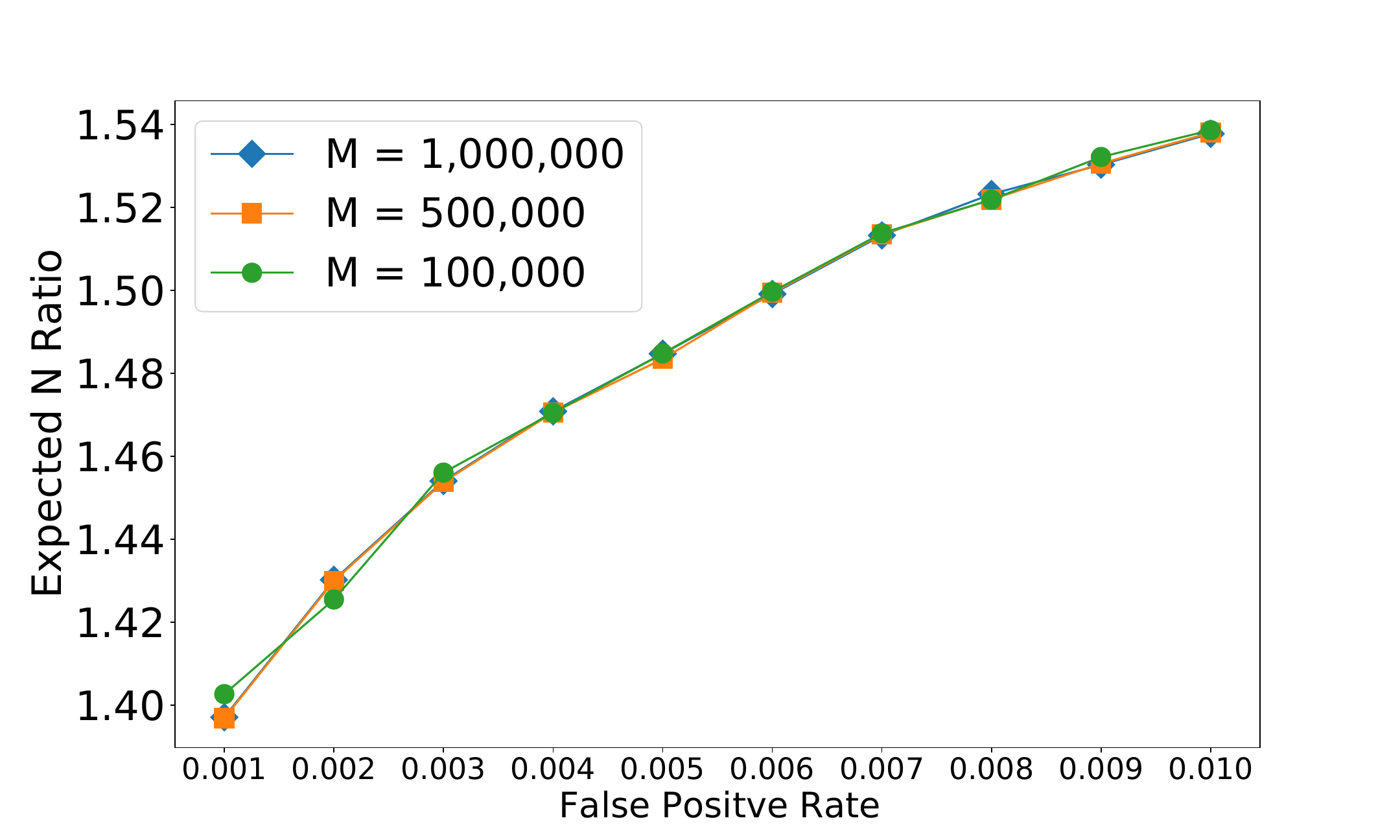}}
\caption{Expected message capacity ratio of a one-phase RBF to a two-phase, for different values of $M$ and False Positive rates.}
\label{fig:oneTwoCompM}
\end{figure}

\section{Related Work}
\label{sec:related}
Luo {\em et al.} provides a general summary of Bloom Filters and their variants, along with their applications to problems in computing \cite{luo2018optimizing}.  To the best of our knowledge, all previous analyses of Bloom Filters approach the question of False Positive rates from the perspective of a ``worst case'' bound.  Equation \eqref{eq:wikiform} is widely used in most applications to compute the False Positive rate according to this metric.  Two prior works  have shown this equation is slightly inaccurate, and in fact a lower bound for the true ``worst case'' False Positive rate \cite{bose2008false,christensen2010new}.  



There have been other BF variants proposed that consider the need to periodically remove items from the filter.  The Counting Bloom Filter (CBF) \cite{851975} and similar variants support deletion of individual elements.  
However, the CBF and its variants all come with increased space and algorithmic overhead to support deletion. Among similar lines, variants such as the Deletable Bloom Filter \cite{rothenberg2010deletable}, the Ternary Bloom Filter \cite{lim2016ternary} and the Quotient Filter \cite{quotientfilters} support element deletion under restricted circumstances, also at the cost of increased space and algorithmic complexity. 

The idea of periodically resetting a targeted subset of BF bits is first proposed in Donnet \textit{et al.} \cite{donnet2006retouched}; while this method is shown to reduce the overall False Positive rate,  it also comes at the cost of increased algorithmic complexity. More recently, Cuckoo Filters have been used as an extension of Bloom filters. Cuckoo filters use Cuckoo hashing to optimize space utilization. These filters notably offer the ability to delete elements post-insertion without practical overhead \cite{fan2014cuckoo}. However, this requires knowledge of the elements that necessitate deletion, a condition not commonly met in many networking applications. 
In addition, unlike RBFs, Cuckoo filters possess a finite unique message capacity.  This inherently limits the scope of its usage under conditions where an unknown and continuous influx of elements is anticipated, which are common in many networking applications.

Many applications employ the RBF strategy as a low-cost alternative to deal with the need to remove elements periodically. Akamai deploys a two-phase RBF approach to guarantee that any content that ends up being cached in their edge servers has been requested at least twice within a designated time frame \cite{maggs2015algorithmic}. Bloom Filter Routing (BFR) has also been introduced in Information-Centric Networks (ICN) to simplify the process of content discovery across the networks \cite{marandi2017bfr} and in wireless networking where Trindade {\em et al.} have designed \textit{Time Aware Bloom Filter} to only remove specific bits that have not been ``hit'' during a predefined time window \cite{trindade2011hran}.


\section{Conclusion}
\label{sec:concl}

Bloom Filters and their variants are a space-efficient data structure employed widely in all manner of computing applications.  Their space efficiency comes with the tradeoff of potential False Positives, and as such much work has been dedicated towards detailed False Positive analysis.  Yet all this work approaches the question of False Positives from the perspective of a ``worst-case'' bound.  This bound is overly conservative for the majority of applications that use Bloom Filters, as it does not take into account the actual state of the Bloom Filter after each arrival.  In fact, applications that use Bloom Filters often have to periodically ``recycle'' the filter once an allowable number of messages threshold has been exceeded.  In cases such as these, different metrics such as the long-term average False Positive rates across new arrivals may be of more interest than a worst-case bound.

We derive a method to efficiently compute the long-term average False Positive rate of a Bloom Filter that periodically ``recycles'' itself (termed a Recycling Bloom Filter).  We use renewal and Markov models to respectively derive lower bound exact expressions for the long-term average False Positive rates, and apply our model to the standard Recycling Bloom Filter, and a ``two-phase'' variant that is popular in network applications.  We demonstrate that the previous worst-case analysis of False Positives can lead to a reduction in the efficiency of RBFs in certain scenarios.  

\bibliographystyle{abbrv}

\bibliography{ourbib}

\onecolumn{
\appendix

\subsection{Count Instance Most Stringent}
\label{sec:count-instance-proof}
\newtheorem{lma}{Lemma}
\newtheorem{cor}{Corollary}
\newtheorem{claim}{Claim}
\newtheorem{thm}{Theorem}

In this appendix, we show that the false positive where we
consider repeat arrivals is upper bounded by the false positive where
we do not consider repeat arrivals (i.e., the analytical model used in the main body of the paper).  To
be clear, let $A$ be a count of all arrivals, where $N$ of them are
what we call {\em non-repeat}: the first time that the particular
element arrives.  We say an arrival triggers a false positive if:

\begin{itemize}

\item It is non-repeat and doesn't set any bits
\item it is a repeat and did not set any bits the first time it
  arrived as a non-repeat (it definitely doesn't set bits in
  subsequent arrivals, but is again considered a false positive based
  on its behavior when it first arrives).
  \end{itemize}

  Let $F$ be the number of false-positives over both types of arrivals and $B$
  be the number of false-positive non-repeat arrivals.  We are claiming
  that over time:
  $$\frac{B}{N} \ge  {\mathcal H} \frac{F}{A}$$
  where ${\mathcal H}$ is a small constant (discussed below).
  
  First, consider an arrival process of elements, where we index
  elements in the order of their first arrival, calling the $i$th such
  element the $i$th {\em non-repeat}.  We break the arrival process
  into {\em intervals}, where the $i$th interval ends with the arrival
  of the $i$th non-repeat.  The first interval contains one arrival:
  the first element.  The second interval ends with the arrival of a
  second unique element, preceded by repeats of that first non-repeat
  element, etc.

 \subsubsection{Useful R.V.s}

 We define several r.v.s.  
 \begin{itemize}

   \item Define $h_i$ to be an indicator of the $i$th non-repeat
     occurring during a given renewal iteration (i.e., the BF has not
     filled and reset prior to the $i$th non-repeat arriving, such
     that the $i$th interval does not occur during that iteration). 

     Note that if $h_i = 0$, then the $i$th non-repeat never actually
     arrived.  However, one can imagine not resetting the bloom filter
     and observing what {\em would have happened} to subsequent
     arrivals had the bloom filter not been reset.  This is relevant
     for the next three r.v. definitions for purposes of defining them
     as independent from $h_i$.

     \item Define $\gamma_i$ to be an indicator that equals 1 when the
       $i$th interval  happens when it is assumed that the $i-1$st interval
       happened (i.e., $P(\gamma_i=1) = P(h_i=1 | h_{i-1}=1)$.  Note we can define $h_i = \prod_{k=1}^i \gamma_k$.

     \item Define $f_i$ to be an indicator that the $i$th non-repeat
       was (or would have been, if the BF were not reset) a false
       positive.  For the ``would have been'' case, one can envision
       when it comes time to resetting the bloom filter, to allow it
       to continue to run just to explore what would have occurred
       with the remaining non-repeat arrivals, and we let $f_i$
       capture this result.

       \item Define $a_{i,j}$ to be the number of times non-repeat $i$
         appears (or would have appeared) in
 the $j$th interval.  Note that $a_{i,j} = 0$ for
 $i < j$, i.e., the $i$th non-repeat cannot appear in an interval
 before its first arrival, and that $a_{i,i} = 1$, i.e., the $i$th
 non-repeat arrives and the $i$th interval ends.
\end{itemize}

\subsubsection{False Positive Formulae}
\newcommand{\fp}{{\mathcal F}}

Define $\fp_n$ to be the false positive rate ignoring repeats, and
$\fp_r$ to be the false positive rate including repeats.  The false
positive rates for a single iteration of filling the B.F. are:

\begin{eqnarray}
  \fp_n & = & \frac{\sum_i h_i f_i}{\sum_i h_i}\\
  \fp_r & = & \frac{\sum_i \sum_{j \ge i}f_i  h_j a_{i,j}}{\sum_i
              \sum_{j\ge i}   h_j a_{i,j}}
\end{eqnarray}

Renewal theory and the law of sums of expectations gives us that the long term rates are simply:
\begin{eqnarray}
  \fp_n &=& \frac{E[\sum_i h_i f_i]}{E[\sum_i h_i]}  = \frac{\sum_iE[ h_i f_i]}{\sum_i E[h_i]}\\
  \fp_r &=&  \frac{E[\sum_i \sum_{j \ge i} E[f_i  h_j a_{i,j}]}{
              E[\sum_i
              \sum_{j\ge i}  h_j a_{i,j}]}\\
              &=&
              \frac{\sum_i \sum_{j \ge i} E[f_i  h_j a_{i,j}]}{
              \sum_i
              \sum_{j\ge i}  E[h_j a_{i,j}]}
\end{eqnarray}

\subsubsection{Strict increase/decrease results of R.V.s}

\begin{lma}
  \label{lma:monotonic}
  $E[h_i] \ge E[h_{i+1}]$ and $E[f_i] \le E[f_{i+1}]$.
\end{lma}

\begin{proof}
The former follows from the fact that in any sample path, $h_i = 0
\rightarrow h_{i+1} = 0$.  The latter follows from the fact that the
$i+1$st non-repeat arrival arrives at a bloom filter that is no less
filled than what the $i$th non-repeat arrives to, and that the
probability of a non-repeat being a false positive is an increasing
function of the number of bits filled in the bloom filter.
\end{proof}

\begin{claim}
  \label{lma-gamma}
  $E[\gamma_i] \ge E[\gamma_{i+1}]$.
\end{claim}
We have yet to formally prove this result, though simulation indicates it holds.  An intuitive "proof" is presented below, although we admit this is not sufficiently formal to claim it has been rigorously shown.
\begin{proof}
  Define $X_i$ to be the number of bits set in the B.F. after the
  $i-1$st arrival whenever it occurs.  It can be shown (and is
  somewhat intuitive) that $P(X_i > \ell) \le P(X_{i+1} > \ell)$ for
  any $\ell$, such that the latter is more likely to trigger a reset,
  making an $i+1$st interval following an $i$th less likely than an
  $i$th interval following an $i-1$st.
  \end{proof}

 \begin{lma}
   \label{lma:underlying}
   For any underlying i.i.d, distribution, $E[a_{i,j}] \ge
   E[a_{i+1,j}]$
 \end{lma}
\begin{proof}
 The claim says that the later non-repeat arrivals don't appear more
 often (in expectation) during any interval than earlier non-repeat
 arrivals.  This is clearly true for $j=i$ since $E[a_{i+1, i}]=0$.
 It is also clearly true for uniform distributions since, after
 occurring, the expected number of repeats for two elements that have
 already occurred is the same (i.e., for uniform, $E[a_{i,j}] =
 E[a_{i+1,j}]$ whenever $j>i+1$.

 For more general distributions,  
 a sketch of the proof is via sample-path analysis, where we focus on non-repeats $i$ and $i+1$, and consider any sample path in which $a_{i,j} < a_{i+1,j}$.  We will show that for such sample path, there is a 1-to-1 and onto mapping to another sample path in which $a_{i,j} > a_{i+1,j}$ and has a greater or equal likelihood of occurring.  Let $x$ be the $i$th non-repeat and $y$ the $i+1$st non-repeat.

 For the first case, consider when $x$ is less popular than the $y$.  In this case, we map sample path $P$ to an alternate path $P'$ by swapping the first arrivals of $x$ and $y$.  Clearly $P'$ has equal likelihood of $P$, since we simply changed the order of two arrivals in an independently drawn sequence.  Also, since their relative initial arrivals have changed order, the resulting $P'$ satisfies $a_{i,j} > a_{i+1,j}$.  The mapping is clearly a bijection since the inverse operation is to re-apply the swap on the $i$th and $i+1$st non-repeat arrivals, reverting $P'$ to $P$.

 For the case where $x$ is more popular than $y$, we keep the initial arrivals fixed but swap all remaining arrivals.  Since 
 $a_{i,j} < a_{i+1,j}$, there are initially more $y$ than $x$, so the resulting sequence $P'$ has more $x$ than $y$, and since $x$ is more popular, $P'$ has larger likelihood.  Again, this mapping is clearly a bijection, since the reverse mapping is again to again swap the repeat arrivals of $x$ and $y$.

 Since every sample path $P$ where $a_{i,j} < a_{i+1,j}$ can bijectively be mapped to a sample path $P'$ with same or greater likelihood where $a_{i,j} > a_{i+1,j}$, the result holds.
 \end{proof}

 \subsubsection{Independence Results}
 
 Define $g_{i,j} = \prod_{k=i+1}^j \gamma_k$.  This can be thought of as
 an indicator interval $j$ occurring, given interval $i$ occurred.
 
 \begin{lma}
   \label{gamma-lemma}
   The $\gamma_k$ are independent from one another, and  $h_i$ is independent from $g_{i,j}$.
 \end{lma}

 \begin{proof}
   Both of these are is by the definition of the
   $\gamma_k, i < k \le j$ that comprise $g_{i,j}$.  Each $\gamma_k$
   is defined such that its value is set under the assumption that
   $h_{k-1} = 1$, so its likelihood of equalling 1 does not actually
   depend on $h_{k-1}$, hence does not depend on $\gamma_m$ for $m<k$.
 \end{proof}

 \begin{cor}
   $E[h_j] = E[h_i] E[g_{i,j}]$
  \end{cor}

 \begin{lma}
   \label{fi-hi-independence}
   $f_i$ and $h_i$ are independent (for same $i$).
 \end{lma}

 \begin{proof}
$f_i$ is defined as a ``would have happened'' r.v., such that it's
value is unaffected by $h_i$, and $h_i$ depends only on events prior
to the $i$th interval where the $i$th non-repeat first arrives and
determines its false positive ($f_i$) status.
\end{proof}

\begin{lma}
  \label{a-lemma}
  $a_{i,j}$ is independent of $f_i$ and $h_i$, and $g_{i,j}$.
\end{lma}

\begin{proof}
  The $a_{i,j}$ counts the number of repeat arrivals of the $i$th
  element during the $j$th interval when if the $i$th arrival and the
  $j$th interval had occurred.  Defining $a_{i,j}$ in this manner
  makes it agnostic as to whether the $i$th and $j$th intervals
  actually occurred, so it is independent of $h_i$ and $g_{i,j}$.

  Similarly, the
  sequence of arrivals during the $j$th can be drawn from the $j-1$
  previously arrived messages, finishing with an arrival of the $j$th
  non-repeat.  Note this process is unaffected by whether these
  messages are false positives or not.  Such information only affects
  the bits in the B.F., which impacts whether the $j$th interval takes
  place, but not the set of messages that would occur in the $j$th
  interval.
  \end{proof}

\begin{lma}
  Given $j>i, P(h_j = 1 | f_i = 1) = P(h_{j-1} = 1)$.
\end{lma}

This Lemma states that the likelihood of the $j$th interval occuring,
given some previous non-repeat is a false-positive equals the
probability of the $j-1$st arrival occuring.

\begin{proof}

  The lefthand side describes the likelihood of having a $j$th
  interval occur when at least one prior interval (the $i$th) is a
  false positive.  The $i$th first non-repeat arrival being a false
  positive means that it does not set any bits in the BF, such that
  the following $j-i-1$ non-repeat arrivals outcomes would be the same
  for the case where the $i$th non-repeat arrival never happened (due
  to the independence of hash functions).  This means for any sample
  path, that after the $i-1$st interval (if it happens), there are
  $j-i-1$ remaining non-repeats whose arrival can potentially set bits
  to cause $h_j = 0$ (i.e., the $i+1$st through $j-1$st).

  Now consider the right-hand side of the equation, and consider the
  same sample path up to and through the $i-1$st interval.  After the
  completion of the $i-1$st interval (if it happens), there are also
  $j-i-1$ remaining non-repeats whose arrival can potentially set bits
  to cause $h_{j-1} = 0$ (i.e., the $i$th through $j-2$nd).

\end{proof}

\begin{cor}
  $E[f_i h_j] = E[f_i] E[h_{j-1}]$
\end{cor}

\begin{proof}
 Since $f_i$ and $h_j$ are indicators, we have that $E[f_i h_j] =
 P(f_i h_j = 1) = P(h_j = 1 | f_i = 1) P(f_i = 1) = P(h_{j-1}=1)
 P(f_i=1) = E[f_i] E[h_{j-1}]$.
\end{proof}

Define ${\mathcal H} = \max_j E[h_{j-1}] / E[h_j]$, i.e., the largest
expected decay factor between intervals.  Then we have that $E[f_i
h_j] = E[f_i] E[h_{j-1}] = E[f_i] E[h_{j}] (E[h_{j-1}] / E[h_j]) \le
E[f_i] E[h_{j}] {\mathcal H}$.


Renewal theory dictates that, followed by expectation-of-sums rule,
followed by independence of the r.v.s yields::
\begin{eqnarray}
  E[\fp_n] & = & \frac{E[\sum_i h_i f_i]}{E[\sum_i h_i]}\\
& = & \frac{\sum_i E[h_i] E[f_i]}{\sum_i E[h_i]}\\
  E[\fp_r] & = & \frac{E[\sum_i \sum_{j \ge i}f_i  h_j
                 a_{i,j}]}{E[\sum_i \sum_{j\ge i}   h_j a_{i,j}]}\\
           & = & \frac{\sum_i \sum_{j \ge i}E[f_i h_j] E[a_{i,j}]} {\sum_i
              \sum_{j\ge i}  E[ h_j]E[ a_{i,j}]}\\
           & = & \frac{\sum_i \sum_{j \ge i}E[f_i] E[ h_{j-1}] E[a_{i,j}]} {\sum_i
                 \sum_{j\ge i}  E[ h_j]E[ a_{i,j}]}\\
           & \le & {\mathcal H} \frac{\sum_i \sum_{j \ge i}E[f_i] E[ h_{j}] E[a_{i,j}]} {\sum_i
                   \sum_{j\ge i}  E[ h_j]E[ a_{i,j}]}\\
  & = & {\mathcal H} \frac{\sum_i E[f_i] E[ h_{i}] \sum_{j \ge i}
        E[g_{i,j}] E[a_{i,j}]} {\sum_i E[ h_i]
                   \sum_{j\ge i}   E[g_{i,j}] E[ a_{i,j}]} \label{reduction-final}
\end{eqnarray}

Define $\alpha_i = \sum_{j \ge i} g_{i,j} a_{i,j}$.

\begin{lma}
  $E[\alpha_i] > E[\alpha_{i+1}]$.
\end{lma}

\begin{proof}
  We show that the $k$th term of $\alpha_i$ is larger than the $k$th
  term of $\alpha_{i+1}$, i.e.,
  $E[g_{i,k} a_{i,k}] = E[g_{i,k}]E[a_{i,k}]\ge E[g_{i+1,k+1}]E[
  a_{i+1,k}] = E[g_{i+1,k+1} a_{i+1,k}]$.  This follows from Lemma
  \ref{a-lemma} that $a_{i,j}$ and the $g_{i,j}$ are
  independent, by Claim \ref{lma-gamma} that $E[\gamma_i]$ are
  decreasing such that $E[g_{i,k}] \ge E[g_{i+1,k+1}]$ and Claim
  \ref{lma:underlying} that $E[a_{i,j}] \ge E[a_{i+1,j}]$.  Hence
  the infinite sum forming $E[\alpha_i]$ is no less than the infinite
  sum forming $E[\alpha_{i+1}]$.

\end{proof}

\newcommand{\AAA}{{\mathcal A}}

Substituting $H_i = E[h_i], F_i = E[f_i], \AAA_i = E[\alpha_i]$, the
above formulae simplify to:
\begin{eqnarray}
   E[\fp_n] & = & \frac{\sum_i H_i F_i}{\sum_i H_i}\\
  E[\fp_r] & = & \frac{\sum_i H_i \AAA_i F_i}{\sum_i H_i
             \AAA_i}
\end{eqnarray}

\subsubsection{Proving the Inequality}

  \begin{lma}
    \label{lma:basic-inequality}
If $\AAA_1 \ge \AAA_2 > 0$ and $F_2 \ge F_1 > 0$., then $\frac{H_1 F_1 + H_2
F_2}{H_1 + H_2} \ge \frac{\AAA_1 H_1 F_1 + \AAA_2 H_2
F_2}{\AAA_1 H_1 + \AAA_2 H_2}$.
\end{lma}

\begin {proof}
$\AAA_1 - \AAA_2 \ge 0$ and $F_1 \ge  F_2$ yields that $F_2
(\AAA_1 - \AAA_2) \ge F_1 (\AAA_1 - \AAA_2)$, such that $F_2
\AAA_1 + F_1 \AAA_2 \ge F_1 \AAA_1 + F_2 \AAA_2$.  Multiplying
both sides by $H_1 H_2$ and then adding $(H_1)^2 \AAA_1 F_1 +
(H_2)^2 \AAA_2 F_2$ to both sides yields:
\begin{eqnarray}
  \nonumber (H_1 \AAA_1 + H_2 \AAA_2)(H_1 F_1 + H_2 F_2)
\nonumber   & = &
                  (H_1)^2 \AAA_1 F_1 + H_1 H_2 F_1 \AAA_2 + H_1 H_2 F_2 \AAA_1
        + (H_2)^2 \AAA_2 F_2\\
\nonumber  & \ge &
          (H_1)^2 \AAA_1 F_1 + H_1 H_2 F_1
          \AAA_1 +  H_1 H_2 F_2 \AAA_2 +
          (H_2)^2 \AAA_2 F_2\\
  & = & (H_1 + H_2) (H_1 F_1 \AAA_1 + H_2 F_2 \AAA_2)
\end{eqnarray}

Dividing both sides by $(H_1 + H_2)(H_1 \AAA_1 + H_2 \AAA_2)$
yields the result.

\end{proof}

\begin{lma}
  \label{inequality-general}
If for all $i$ we have 
$F_{i+1} \ge F_i$ and $\AAA_{i} \ge \AAA_{i+1}$, then $\frac{\sum_{i=1}^m H_i F_i}{
\sum_{i=1}^m H_i} \ge \frac{\sum_{i=1}^m \AAA_i H_i F_i}{\sum_{i=1}^{m} \AAA_i H_i}$.
\end{lma}

\begin{proof}
  This can be proven inductively on the number of terms in the sum,
  and taking the limit of the number of terms to $\infty$.

For the base case,   Lemma
\ref{lma:basic-inequality} applies directly when $m=2$.  For larger
$m$, we define the following:

\begin{eqnarray}
  \label{barF-def}
  \bar{F} & = & \frac{\sum_{i=1}^{m-1} H_i F_i}{\sum_{i=1}^{m-1} H_i}\\
  \beta & = & \frac{\sum_{i=1}^{m-1} H_i F_i \AAA_i}{\sum_{i=1}^{m-1}
              H_i F_i}\\
  \bar{\AAA} & = & \frac{\sum_{i=1}^{m-1} H_i \AAA_i}{\sum_{i=1}^{m-1}
               H_i }\\
\end{eqnarray}
Also note that  $F_i$ increasing yields
$\bar{F} =  \frac{\sum_{i=1}^{m-1} H_i F_i}{\sum_{i=1}^{m-1}
  H_i} \le \frac{\sum_{i=1}^{m-1} H_i F_{m-1}}{\sum_{i=1}^{m-1} H_i}
= F_{m-1} \frac{\sum_{i=1}^{m-1} H_i}{\sum_{i=1}^{m-1} H_i} \le F_m$ and that because we have $\AAA_i$ decreasing, we similarly have
$\bar{\AAA} = \frac{\sum_{i=1}^{m-1} H_i F_i
  \AAA_{i}}{\sum_{i=1}^{m-1} H_i F_i} \ge \frac{\sum_{i=1}^{m-1} H_i F_i
  \AAA_{m-1}}{\sum_{i=1}^{m-1} H_i F_i} \ge \AAA_m$.

Finally, we make our inductive assumption that 
$$
\frac{\sum_{i=1}^{m-1} H_i F_i \AAA_i}{\sum_{i=1}^{m-1} H_i
\AAA_i} \le \frac{\sum_{i=1}^{m-1} H_i
F_i}{\sum_{i=1}^{m-1} H_i}.$$

We can rewrite the left-hand side as
$\frac{\sum_{i=1}^{m-1} H_i F_i \AAA_i}{\sum_{i=1}^{m-1} H_i
\AAA_i}  = \frac{\beta \sum_{i=1}^{m-1} H_i F_i}{\bar{\AAA}
  \sum_{i=1}^{m-1} H_i} $, and this being less than the right hand
side yields $\beta \le \bar{\AAA}$.

This yields:
\begin{eqnarray}
\nonumber \frac{(\sum_{i=1}^{m-1} H_i F_i \AAA_i)+ H_m F_m \AAA_m}{(\sum_{i=1}^{m-1} H_i
\AAA_i) + H_m \AAA_m} &= &   \frac{(\beta \sum_{i=1}^{m-1} H_i F_i) + H_m
F_m \AAA_m}{(\bar{\AAA} \sum_{i=1}^{m-1}  H_i)
+ H_m \AAA_m} \le \frac{(\bar{\AAA} \sum_{i=1}^{m-1} H_i F_i) + H_m
F_m \AAA_m}{(\bar{\AAA} \sum_{i=1}^{m-1} H_i)
+ H_m \AAA_m}\\
\label{ineq-pre}
& = & \frac{(\bar{\AAA} \bar{F} \sum_{i=1}^{m-1} H_i) + H_m F_m \AAA_m}{(\bar{\AAA}
  \sum_{i=1}^{m-1}H_i) + H_m \AAA_m} 
\end{eqnarray}

We can apply Lemma \ref{lma:basic-inequality} to this final outcome,
using $\bar{F}$ and $F_m$ respectively as the $F_1$ and $F_2$ of Lemma
\ref{lma:basic-inequality},
$\sum_{i=1}^{m-1}H_i$ and $H_m$ respectively as the $H_1$ and $H_2$ of
Lemma \ref{lma:basic-inequality}, and
$\bar{\AAA}$  and $\AAA_m$ respectively as the $\AAA_1$ and $\AAA_2$
of Lemma \ref{lma:basic-inequality} to
give:

\begin{eqnarray}
\label{ineq1}
\frac{(\sum_{i=1}^{m} H_i F_i \AAA_i)}{(\sum_{i=1}^{m} H_i
\AAA_i)}  & = & \frac{(\sum_{i=1}^{m-1} H_i F_i \AAA_i)+ H_m F_m \AAA_m}{(\sum_{i=1}^{m-1} H_i
                \AAA_i) + H_m \AAA_m} 
   \le  \frac{(\bar{\AAA} \bar{F} \sum_{i=1}^{m-1} H_i) + H_m F_m \AAA_m}{(\bar{\AAA}
  \sum_{i=1}^{m-1}H_i) + H_m \AAA_m} \\
  \label{ineq2}
& \le & \frac{\bar{F}
        \sum_{i=1}^{m-1} H_i + H_m F_m}{\sum_{i=1}^{m-1}H_i + H_m}
        = \frac{
        \sum_{i=1}^{m-1} H_i F_i + H_m F_m}{\sum_{i=1}^{m-1}H_i +
        H_m}\\
  & = & \frac{
        \sum_{i=1}^{m} H_i F_i}{\sum_{i=1}^{m}H_i}
\end{eqnarray}
where (\ref{ineq1}) follows from (\ref{ineq-pre}), the inequality of
(\ref{ineq2}) follows via application of Lemma
\ref{lma:basic-inequality}, and the equality of (\ref{ineq2}) follows
via substitution of (\ref{barF-def}).

\end{proof}

\begin{thm}
  $E[\fp_n] \ge E[\fp_r]/ {\mathcal H}$
\end{thm}

\begin{proof}
  Follows from inequality \ref{reduction-final} and Lemma
  \ref{inequality-general}, and applying the limit as $m \rightarrow \infty$.
\end{proof}

\subsubsection{Thoughts on ${\mathcal H}$}

Note the upper bound is thinned by a factor of $1 / {\mathcal H}$ where
$1 / {\mathcal H} = \min_j E[h_{j}] / E[h_{j-1}]$, where for each $j$, 
$E[h_{j}] / E[h_{j-1}] = P(h_j=1) /
P(h_{j-1}=1)$.  Since $h_j = 1 \rightarrow h_{j-1} = 1$, this equals
$P(h_j=1, h_{j-1} = 1) / P(h_{j-1}=1) = P(h_j=1 | h_{j-1} = 1) =
\gamma_j$.  Due to the variance in the number of bits set after the
$j$th non-repeat, it is unlikely that any particular $\gamma_j$ will
be significantly less than 1: this can be verified empricially (the
underlying distribution does not matter).

}


\end{document}